\documentclass[a4paper,11pt]{article}

\usepackage{amsfonts}
\usepackage[dvips]{color,graphicx}
\usepackage{graphicx,amssymb,amsmath,bm,latexsym}
\usepackage{stmaryrd}
\usepackage{bm}

\newtheorem{theorem}{Theorem}

\newtheorem{definition}[theorem]{Definition}

\newtheorem{proposition}[theorem]{Proposition}

\newenvironment{proof}[1][Proof]{\noindent\textbf{#1.} }{\ \hfill\rule{0.5em}{0.5em}}

\definecolor{red  }{rgb}{1,0,0}
\definecolor{blue }{rgb}{0,0,1}
\definecolor{green}{rgb}{0,1,0}

\newcommand{\vs}[1]{\vspace{#1 mm}}

\usepackage{graphicx,amssymb,amsmath,bm,latexsym}

\begin{document}

\begin{titlepage}

\vskip .5in

\begin{center}

{\Large\bf On $q$-deformed infinite-dimensional $n$-algebra } \vskip .5in

{\large Lu Ding$^{a}$,
Xiao-Yu Jia$^{b}$,
Ke Wu$^{b,c}$,
Zhao-Wen Yan$^{d}$,
Wei-Zhong Zhao$^{b,c,}$
\footnote{Corresponding author: zhaowz@cnu.edu.cn}} \\
\vs{10}

$^a${\em Institute of Applied Mathematics, Academy of Mathematics and Systems Science, Chinese
Academy of Sciences, Beijing 100190, China} \\
$^b${\em School of Mathematical Sciences, Capital Normal University,
Beijing 100048, China} \\
$^c${\em Beijing Center for Mathematics and Information
Interdisciplinary Sciences, Beijing 100048, China } \\
$^d${\em School of Mathematical Sciences, Inner Mongolia University,
Hohhot 010021, China}\\

\vskip .2in \vspace{.3in}

\begin{abstract}
The $q$-deformation of the infinite-dimensional $n$-algebra is investigated.
Based on the structure of the $q$-deformed Virasoro-Witt algebra,
we derive a nontrivial $q$-deformed Virasoro-Witt $n$-algebra
which is nothing but a sh-$n$-Lie algebra.
Furthermore in terms of the pseud-differential operators on the quantum plane,
we construct the (co)sine $n$-algebra and the $q$-deformed $SDiff(T^2)$ $n$-algebra.
We prove that they are the sh-$n$-Lie algebras for the case of even $n$.
An explicit physical realization of the (co)sine $n$-algebra is given.

\end{abstract}

\end{center}

{\small KEYWORDS: $q$-deformation, Conformal and $W$ Symmetry, $n$-algebra }


\vfill

\end{titlepage}

\section{Introduction}

Quantum algebras or more precisely quantized universal enveloping algebras
first appeared in connection with the study of the inverse scattering
problem. It is one parameter $q$ deformation of Lie algebras which preserves
the structure of a Hopf algebra and reduces to standard Lie algebra in the
classical limit. The Virasoro algebra is an infinite dimensional Lie algebra
and plays important roles in physics. Its $q$-deformation has been widely
studied in the literature \cite{Curtright1990}-\cite{Sato}. A $q$-deformation
of the centerless Virasoro or Virasoro-Witt (V-W) algebra was first obtained
by Curtright and Zachos \cite{Curtright1990}. Its central extension was
later furnished by Aizawa and Sato \cite{Aizawa}. Chaichian and
Pre$\check{s}$najder \cite{Chaichian1992} proposed a different version of the $q$-deformed
Virasoro algebra by carrying out a Sugawara construction on a $q$-analogue of
an infinite dimensional Heisenberg algebra. Shiraishi et al. \cite{Shiraishi}
presented a $q$-Virasoro algebra $\mathit{Vir}_{q,t}$, where $q$ and $t$ are
two complex parameters. They constructed a free boson realization of this
$q$-Virasoro algebra and showed that singular vectors can be expressed by the
Macdonald symmetric functions. It is similar as the case of the ordinary
Virasoro algebra whose singular vectors are given by the Jack symmetric
functions. It is well-known that there is a remarkable connection between
the Virasoro algebra and the Korteweg-de Vries (KdV) equation \cite%
{Gervais1, Gervais2}. For the $q$-deformed Virasoro algebra, Chaichian et al.
\cite{Chaichian249} showed that it generates the sympletic structure which
can be used for a description of the discretization of the KdV equation.
Furthermore the quantum KdV equations associated with the algebraic symmetry
have been investigated in Refs.\cite{Bazhanov, Fioravanti}. The integrable
one-dimensional quantum spin chains have attracted much interest from
physical and mathematical points of view. One noted that the deformed
Virasoro algebra plays an important role in the study of the XYZ model \cite%
{Lukyanov, Rossi}.

The $W_{\infty}$ algebra is the higher-spin extension of the Virasoro algebra.
Its $q$-deformation has been constructed \cite{Chaichian237, Zha}.
The classical limit of the $W_{\infty}$ algebra gives the so-called $w_{\infty}$ algebra
which is equivalent to the algebra of smooth area-preserving diffeomorphisms
of the cylinder $S^1\times R^1$ \cite{Pope236}.
It is worth to emphasize that the algebra of the area-preserving diffeomorphisms
of the torus $T^2$, i.e., $SDiff(T^2)$ \cite{Floratos201,Antoniadis300} is also an
important infinite-dimensional algebra.
In terms of the Gauss derivatives on the quantum plane,
Kinani et al. \cite{Kinani357} presented the $q$-deformed $SDiff(T^2)$ algebra.
It should be noted that the sine algebra  arises as the unique
Lie algebra deformation of $SDiff(T^2)$ in some suitable
basis. There has been considerable interest in the (super) sine algebra
\cite{Fairlie224}-\cite{Jellal}.

The Nambu 3-algebra was introduced in \cite{Nambu, Takhtajan} as a natural
generalization of a Lie algebra for higher-order algebraic operations.
Recently Bagger and Lambert \cite{BL2007, BL2008}, and Gustavsson \cite%
{Gustavsson} (BLG) found that 3-algebras play an important role in
world-volume description of multiple M2-branes. Due to BLG theory, there has
been considerable interest in the 3-algebra and its application. More
recently there has been the progress in constructing the
infinite-dimensional 3-algebras, such as V-W \cite{Curtright, Curtright2009}%
, Kac-Moody \cite{Lin} and $w_{\infty}$ 3-algebras \cite{Chakrabortty, Chen}%
. Moreover the relation between the infinite-dimensional 3-algebras and the
integrable systems has also been studied. \cite{Chen2012, Chen2013}.

Recently Curtright et al. \cite{Curtright}, constructed a V-W algebra through the use
of su(1,1) enveloping algebra techniques. It is worthwhile to mention that
this ternary algebra depends on a parameter $z$ and is only a Nambu-Lie
algebra when $z=\pm 2i$. Ammar et al. \cite{Ammar} presented a $q$-deformation
of this 3-algebra and noted it carrying the structure of ternary
Hom-Nambu-Lie algebra.
It is well known that the structure and property of $q$-deformed algebra are now very well
understood. However for the $q$-deformed infinite-dimensional 3-algebra,  much less is still known
about its structure and property. As to the $q$-deformed infinite-dimensional $n$-algebra,
to our best knowledge, it has not been reported so far in the existing literature.
The goal of this paper is to construct the $q$-deformed infinite-dimensional $n$-algebras
and explore their intriguing features.

This paper is organized as follows. In section 2, we introduce the definitions of $n$-Lie algebra
and sh-$n$-Lie algebra.
In section 3, we construct the $q$-deformed V-W $n$-algebra.
In section 4, in terms of the pseud-differential operators on the quantum plane,
we construct the (co)sine $n$-algebra and the $q$-deformed $SDiff(T^2)$ $n$-algebra.
An explicit physical realization of the (co)sine $n$-algebra is given in section 5.
We end this paper with the concluding remarks in section 6.


\section{ $n$-Lie algebra and sh-$n$-Lie algebra}

For later convenience, we shall recall the definitions of $n$-Lie algebra
and sh-$n$-Lie algebra in this section. For a more detailed description we
refer the reader to Refs.\cite{Filippov}-\cite{Izquierdo}.

The notion of $n$-Lie algebra or Filippov $n$-algebra was introduced by
Filippov \cite{Filippov}. It is a natural generalization of Lie algebra.

\begin{definition}
\cite{Filippov}An $n$-Lie algebra structure is a linear space $V$ endowed
with a multilinear map called Nambu bracket $[.,\cdots ,.]$: $V^{\otimes
n}\rightarrow V$ satisfying the following properties:\newline
(1). Skew-symmetry
\begin{equation}
\lbrack A_{\sigma (1)},\cdots ,A_{\sigma (n)}]=(-1)^{\epsilon (\sigma
)}[A_{1},\cdots ,A_{n}].
\end{equation}%
(2). Fundamental identity (FI) or Filippov condition
\begin{equation}
\lbrack A_{1},\cdots ,A_{n-1},[B_{1},\cdots
,B_{n}]]=\sum_{k=1}^{n}[B_{1},\cdots ,B_{k-1},[A_{1},\cdots
,A_{n-1},B_{k}],B_{k+1},\cdots ,B_{n}].  \label{eq:eFI}
\end{equation}
\end{definition}

Recently 3-Lie algebra has attracted much interest due to its application in
M-theory. For the case of 3-Lie algebra, the corresponding FI is
\begin{equation}
\lbrack A,B,[C,D,E]]=[[A,B,C],D,E]+[C,[A,B,D],E]+[C,D,[A,B,E]].
\label{eq:FI}
\end{equation}%
We have already seen that an $n$-Lie algebra $A$ is a vector space $A$
endowed with an $n$-ary skew-symmetric multiplication satisfying the
FI (\ref{eq:eFI}). We now turn to the notion of sh-$n$-Lie algebra.

\begin{definition}
\cite{Goze}Let $[.,\cdots ,.]$ be an $n$-ary skewsymmetric product on a
vector space $V$. We say that $(V,[.,\cdots ,.])$ is a sh-$n$-Lie algebra if
$[.,\cdots ,.]$ satisfies the sh-Jacobi's identity
\begin{equation}
\sum_{\sigma \in Sh(n,n-1)}(-1)^{\epsilon (\sigma )}\left[ [x_{\sigma
(1)},\cdots ,x_{\sigma (n)}],x_{\sigma (n+1)},\cdots ,x_{\sigma (2n-1)}%
\right] =0,  \label{eq:shJacobi}
\end{equation}%
for any $x_{i}\in A$, where $Sh(n,n-1)$ is the subset of $\Sigma _{2n-1}$
defined by
\begin{equation*}
Sh(n,n-1)=\{\sigma \in \Sigma _{2n-1},\sigma (1)<\cdots <\sigma (n),\sigma
(n+1)<\cdots <\sigma (2n-1)\}.
\end{equation*}
\end{definition}

In terms of the L$\acute{e}$vi-Civit$\grave{a}$ symbol, i.e.,
\begin{equation}
\epsilon_{j_{1}\cdots j_{p}}^{i_{1}\cdots i_{p}}=\det\left(
\begin{array}{ccc}
\delta_{j_{1}}^{i_{1}} & \cdots & \delta_{j_{p}}^{i_{1}} \\
\vdots &  & \vdots \\
\delta_{j_{1}}^{i_{p}} & \cdots & \delta_{j_{p}}^{i_{p}}%
\end{array}
\right) ,  \label{eq:LV}
\end{equation}
the sh-Jacobi's identity (\ref{eq:shJacobi}) can also be expressed as
\begin{equation}
\epsilon_{m_{1}\cdots m_{2n-1}}^{i_{1}\cdots i_{2n-1}}\left[
[x_{i_{1}},\cdots,x_{i_{n}}],x_{i_{n+1}},\cdots,x_{i_{2n-1}}\right] =0.
\label{emisila sh}
\end{equation}

When $n=2$, both the FI (\ref{eq:eFI}) and sh-Jacobi's
identity (\ref{eq:shJacobi}) become the well-known Jacobi's identity. When $%
n=3$, the FI is given by (\ref{eq:FI}). For this case, the
corresponding sh-Jacobi's identity (\ref{eq:shJacobi}) is
\begin{eqnarray}
&&[[A,B,C],D,E]-[[A,B,D],C,E]+[[A,B,E],C,D]+[[A,C,D],B,E]  \notag
\label{eq:shuffle} \\
&&-[[A,C,E],B,D]+[[A,D,E],B,C]-[[B,C,D],A,E]+[[B,C,E],A,D]  \notag \\
&&-[[B,D,E],A,C]+[[C,D,E],A,B]=0.
\end{eqnarray}

We have briefly introduced the $n$-Lie algebra and sh-$n$-Lie algebra. It
should be noted that any $n$-Lie algebra is a sh-$n$-Lie algebra, but a sh-$%
n $-Lie algebra is an $n$-Lie algebra if and only if any adjoint operator is
a derivation.

\section{ $q$-deformed V-W $n$-algebra}
\subsection{ $q$-deformed V-W 3-algebra}
As a start before investigating the $q$-deformed 3-algebra, let us recall the
case of $q$-deformed algebra. The deformation of the commutator is defined by
\begin{equation}  \label{eq:dc}
[A, B]_{(p,q)}=pAB-qBA.
\end{equation}
It possesses the following properties \cite{Chaichian248, Sato}:
\begin{eqnarray}
&&[A, B]_{(p,q)}=-[B, A]_{(q,p)},  \notag \\
&&[A+B, C]_{(p,q)}=[A, C]_{(p,q)}+[B, C]_{(p,q)},  \notag \\
&&[AB, C]_{(p,q)}=A[B, C]_{(p,r)}+[A, C]_{(r,q)}B,  \notag \\
&&[A, BC]_{(p,q)}=B[A, C]_{(r,q)}+[A, B]_{(p,r)}C,
\end{eqnarray}
and the $q$-Jacobi identity
\begin{eqnarray}
&&[A, [B, C]_{(q_1,q_1^{-1})}]_{(q_3/q_2, q_2/q_3)}+[B, [C,
A]_{(q_2,q_2^{-1})}]_{(q_1/q_3, q_3/q_1)}  \notag \\
&&+[C, [A, B]_{(q_3,q_3^{-1})}]_{(q_2/q_1, q_1/q_2)}=0.
\end{eqnarray}

The Virasoro algebra is an infinite dimensional Lie algebra and plays
important roles in physics. The V-W algebra is indeed the centerless
Virasoro algebra. It is given by
\begin{equation}
\lbrack L_{m},L_{n}]=(m-n)L_{m+n}.  \label{eq:virwitt}
\end{equation}%

To construct the deformed V-W algebra, let us take the $q$-deformed generators
\begin{equation}  \label{eq:generator}
L_m=-q^{N}(a^\dag)^{m+1}a,
\end{equation}
where the $q$-deformed oscillator is deformed by the following relations \cite%
{MacFarlane}-\cite{Hayashi}:
\begin{eqnarray}  \label{eq:oscillator}
&&a a^\dag-qa^\dag a=q^{-N},\ \ a a^\dag=[N],  \notag \\
&&[N,a]=-a,\ \ [N,a^\dag]=a^\dag.
\end{eqnarray}
Substituting the $q$-generators (\ref{eq:generator}) into the commutator (\ref%
{eq:dc}) and using the $q$-deformed oscillator (\ref{eq:oscillator}), it leads
to the so-called $q$-deformed V-W algebra \cite{Curtright1990}
\begin{eqnarray}  \label{eq:qvir}
[L_m, L_n]_{(q^{m-n}, q^{n-m})}=q^{m-n}L_mL_n-q^{n-m}L_nL_m =[m-n]L_{m+n},
\end{eqnarray}
where $[k]=\frac{q^k-q^{-k}}{q-q^{-1}}$. In the limit $q\rightarrow 1$, (\ref%
{eq:qvir}) reduces to the V-W algebra (\ref{eq:virwitt})

Let us define the star product by
\begin{eqnarray}  \label{eq:qpro}
&&L_n*[L_m, L_k]_{(q^{m-k}, q^{k-m})}=q^{2n-m-k}L_n[L_m, L_k]_{(q^{m-k},
q^{k-m})},  \notag \\
&&[L_m, L_k]_{(q^{m-k}, q^{k-m})}*L_n=q^{m+k-2n}[L_m, L_k]_{(q^{m-k},
q^{k-m})}L_n.
\end{eqnarray}
Then we have
\begin{eqnarray}  \label{eq:cqj}
&&L_n*[L_m, L_k]_{(q^{m-k}, q^{k-m})}-[L_m, L_k]_{(q^{m-k}, q^{k-m})}*L_n
\notag \\
&&=[L_n, [L_m, L_k]_{(q^{m-k}, q^{k-m})}]_{(q^{2n-m-k}, q^{m+k-2n})}.
\end{eqnarray}
By means of (\ref{eq:cqj}), one can confirm the following $q$-Jacobi identity
\cite{Chaichian248} satisfied by the $q$-deformed V-W algebra (\ref{eq:qvir}):
\begin{eqnarray}  \label{eq:qji}
[L_n, [L_m, L_k]_{(q^{m-k}, q^{k-m})}]_{(q^{2n-m-k}, q^{m+k-2n})}+cycl.
perms.=0.
\end{eqnarray}

Let us now turn our attention to the case of 3-algebra. The operator Nambu
3-bracket is defined to be a sum of single operators multiplying commutators
of the remaining two \cite{Nambu}, i.e.,
\begin{eqnarray}  \label{eq:3bracket}
[A, B, C]=A[B, C]+B[C, A]+C[A, B],
\end{eqnarray}
where $[A, B]=AB-BA$.

For the $q$-deformed V-W algebra (\ref{eq:qvir}), we have already seen that
the $q$-Jacobi identity (\ref{eq:qji}) is guaranteed to hold. It is worth to
emphasize that the star product (\ref{eq:qpro}) plays a pivotal role in the
$q$-Jacobi identity. In terms of the star product (\ref{eq:qpro}), let us
define the $q$-3-bracket as follows:
\begin{eqnarray}
{\llbracket}L_{m},L_{n},L_{k}{\rrbracket} &=&L_{m}\ast \lbrack
L_{n},L_{k}]_{(q^{n-k},q^{k-n})}+L_{n}\ast \lbrack
L_{k},L_{m}]_{(q^{k-m},q^{m-k})}  \notag  \label{eq:q3bracket} \\
&+&L_{k}\ast \lbrack L_{m},L_{n}]_{(q^{m-n},q^{n-m})}.
\end{eqnarray}%
By means of (\ref{eq:qvir}) and (\ref{eq:qpro}), we may derive the following
$q$-deformed 3-algebra from (\ref{eq:q3bracket}):
\begin{eqnarray}
{\llbracket}L_{m},L_{n},L_{k}{\rrbracket} &=&\frac{1}{q-q^{-1}}%
([2m-2k]+[2k-2n]+[2n-2m])L_{m+n+k}  \notag \\
&=&\left( q-q^{-1}\right) ([m-n][m-k][n-k])L_{m+n+k}  \notag \\
&=&-\frac{1}{\left( q-q^{-1}\right) ^{2}}\det \left(
\begin{array}{ccc}
q^{-2m} & q^{-2n} & q^{-2k} \\
1 & 1 & 1 \\
q^{2m} & q^{2n} & q^{2k}%
\end{array}%
\right) L_{m+n+k}.  \label{eq:qv3alg}
\end{eqnarray}

Performing lengthy but straightforward calculations, we find that (\ref%
{eq:qv3alg}) satisfies the sh-Jacobi's identity (\ref{eq:shuffle}), but the
FI (\ref{eq:FI}) does not hold. It is easy to verify that
the skew-symmetry holds for this ternary algebra
\begin{equation*}
{\llbracket}L_{m},L_{n},L_{k}{\rrbracket}=-{\llbracket}L_{n},L_{m},L_{k}{%
\rrbracket}=-{\llbracket}L_{k},L_{n},L_{m}{\rrbracket}.
\end{equation*}%
Therefore the $q$-deformed V-W 3-algebra (\ref{eq:qv3alg}) is indeed a
sh-3-Lie algebra. In the limit $q\rightarrow 1$, (\ref{eq:qv3alg}) reduces
to the null 3-algebra derived in \cite{Curtright2009},
\begin{equation*}
\lbrack L_{m},L_{n},L_{k}]=0.
\end{equation*}%
The FI (\ref{eq:FI}) is trivially satisfied for this null
3-algebra.

\subsection{$q$-deformed V-W $n$-algebra}
Now encouraged by the possibility of constructing the nontrivial sh-3-Lie
algebra (\ref{eq:qv3alg}), it would be interesting to study further and see
whether one could construct the $q$-deformed V-W $n$-algebra with a genuine sh-%
$n$-Lie algebra structure. In this section we give affirmative answer to
this question.

The $n$-bracket with $n\geq3$ is defined by
\begin{eqnarray}  \label{eq:lbracket}
[L_{i_{1}}, L_{i_{2}}, \cdots, L_{i_{n}}]=\sum\limits_{s=1}^{n}\left(
-1\right) ^{s+1}L_{i_{s}}\lbrack L_{i_{1}}, L_{i_{2}}, \cdots, \widehat{%
L_{i_{s}}}, \cdots, L_{i_{n}}].
\end{eqnarray}
Here we denote a notational convention used frequently in the rest of this
paper. Namely for any arbitrary symbol $Z$, the hat symbol $\hat{Z}$ stands
for the term that is omitted.

Let us define a $q$-$n$-bracket as follows:
\begin{equation}
{\llbracket}L_{i_{1}},L_{i_{2}},\cdots ,L_{i_{n}}{\rrbracket}%
=\sum\limits_{s=1}^{n}\left( -1\right) ^{s+1}L_{i_{s}}\ast {\llbracket}%
L_{i_{1}},L_{i_{2}},\cdots ,\widehat{L_{i_{s}}},\cdots ,L_{i_{n}}{\rrbracket}%
,  \label{eq:qlbracket}
\end{equation}%
where the general star product is given by
\begin{equation}
L_{i_{1}}\ast {\llbracket}L_{i_{2}},L_{i_{3}},\cdots ,L_{i_{n}}{\rrbracket}%
=q^{xi_{1}+y\left( i_{2}+\cdots +i_{n}\right) }L_{i_{1}}{\llbracket}%
L_{i_{2}},L_{i_{3}},\cdots ,L_{i_{n}}{\rrbracket},  \label{eq:gsp}
\end{equation}%
in which $(x=2,y=-1)$ for $n=3$, $(x=n,y=0)$ for even $n\geq 4$ and $%
(x=n-1,y=-2)$ for odd $n\geq 5$. As done in the case of $q$-3-bracket (\ref%
{eq:q3bracket}), we introduce the general star product (\ref{eq:gsp}) into
the $q$-$n$-bracket here. It should be noted that the general star product (%
\ref{eq:gsp}) will play an important role in deriving the desired $q$-deformed
V-W $n$-algebra.

\begin{theorem}
The $q$-generators (\ref{eq:generator}) satisfy the following closed algebraic
structure relation:
\end{theorem}

\begin{align}
& {\llbracket}L_{i_{1}},L_{i_{2}},\cdots,L_{i_{n}}{\rrbracket}=\frac {%
sign\left( n\right) }{\left( q-q^{-1}\right) ^{n-1}}  \notag
\label{eq:qlrbracket} \\
& \det\left(
\begin{array}{cccc}
q^{-2\left\lfloor \frac{n-1}{2}\right\rfloor i_{1}} & q^{-2\left\lfloor
\frac{n-1}{2}\right\rfloor i_{2}} & \cdots & q^{-2\left\lfloor \frac{n-1}{2}%
\right\rfloor i_{n}} \\
q^{2\left( -\left\lfloor \frac{n-1}{2}\right\rfloor +1\right) i_{1}} &
q^{2\left( -\left\lfloor \frac{n-1}{2}\right\rfloor +1\right) i_{2}} & \cdots
& q^{2\left( -\left\lfloor \frac{n-1}{2}\right\rfloor +1\right) i_{n}} \\
\vdots & \vdots & \vdots & \vdots \\
q^{2\left( \left\lfloor \frac{n}{2}\right\rfloor -1\right) i_{1}} &
q^{2\left( \left\lfloor \frac{n}{2}\right\rfloor -1\right) i_{2}} & \cdots &
q^{2\left( \left\lfloor \frac{n}{2}\right\rfloor -1\right) i_{n}} \\
q^{2\left\lfloor \frac{n}{2}\right\rfloor i_{1}} & q^{2\left\lfloor \frac {n%
}{2}\right\rfloor i_{2}} & \cdots & q^{2\left\lfloor \frac{n}{2}%
\right\rfloor i_{n}}%
\end{array}
\right) L_{\Sigma_{l=1}^{n}i_{l}},
\end{align}
where $\left\lfloor n\right\rfloor =Max\{{m\in\mathbf{Z}|m\leq n\}}$ is the
floor function, $sign\left( n\right) $ is the signature function, i.e., $%
sign\left( n\right) =\left\{
\begin{array}{c}
1,\text{ for }n\text{ mod }4=0,1 \\
-1,\text{ for }n\text{ mod }4=2,3%
\end{array}
.\right. $

\begin{proof}
Let us confirm this by the mathematical induction for $n$. Equation (\ref%
{eq:qv3alg}) indicates that (\ref{eq:qlrbracket}) is satisfied for $n=3$. We
suppose (\ref{eq:qlrbracket}) is satisfied for $n$. By means of (\ref%
{eq:qlbracket}), we have
\begin{align}
& {\llbracket}L_{i_{1}},L_{i_{2}},\cdots ,L_{i_{n+1}}{\rrbracket}  \notag \\
& =\frac{sign\left( n\right) }{\left( q-q^{-1}\right) ^{n-1}}A\left[
q^{-\Sigma _{j=1}^{n+1}i_{j}-1}q^{2N}\left( a^{+}\right) ^{\Sigma
_{j=1}^{n+1}i_{j}+2}a^{2}-\frac{1}{\left( q-q^{-1}\right) }L_{\Sigma
_{j=1}^{n+1}i_{j}}\right] +\frac{sign\left( n\right) }{\left(
q-q^{-1}\right) ^{n}}  \notag \\
& \det \left(
\begin{array}{ccccc}
q^{xi_{1}} & \cdots & q^{xi_{s}} & \cdots & q^{xi_{n+1}} \\
q^{\left( y-2\left\lfloor \frac{n-1}{2}\right\rfloor \right) i_{1}} & \cdots
& q^{\left( y-2\left\lfloor \frac{n-1}{2}\right\rfloor \right) i_{s}} &
\cdots & q^{\left( y-2\left\lfloor \frac{n-1}{2}\right\rfloor \right)
i_{n+1}} \\
q^{\left( y-2\left\lfloor \frac{n-1}{2}\right\rfloor +2\right) i_{1}} &
\cdots & q^{\left( y-2\left\lfloor \frac{n-1}{2}\right\rfloor +2\right)
i_{s}} & \cdots & q^{\left( y-2\left\lfloor \frac{n-1}{2}\right\rfloor
+2\right) i_{n+1}} \\
\vdots & \vdots & \vdots & \vdots & \vdots \\
q^{\left( y+2\left\lfloor \frac{n}{2}\right\rfloor \right) i_{1}} & \cdots &
q^{\left( y+2\left\lfloor \frac{n}{2}\right\rfloor \right) i_{s}} & \cdots &
q^{\left( y+2\left\lfloor \frac{n}{2}\right\rfloor \right) i_{n+1}}%
\end{array}%
\right) L_{\Sigma _{l=1}^{n+1}i_{l}},  \label{eq:det2}
\end{align}%
where $A=\det \left(
\begin{array}{ccccc}
q^{xi_{1}} & \cdots & q^{xi_{s}} & \cdots & q^{xi_{n+1}} \\
q^{\left( y-2\left\lfloor \frac{n-1}{2}\right\rfloor +2\right) i_{1}} &
\cdots & q^{\left( y-2\left\lfloor \frac{n-1}{2}\right\rfloor +2\right)
i_{s}} & \cdots & q^{\left( y-2\left\lfloor \frac{n-1}{2}\right\rfloor
+2\right) i_{n+1}} \\
q^{\left( y-2\left\lfloor \frac{n-1}{2}\right\rfloor +4\right) i_{1}} &
\cdots & q^{\left( y-2\left\lfloor \frac{n-1}{2}\right\rfloor +4\right)
i_{s}} & \cdots & q^{\left( y-2\left\lfloor \frac{n-1}{2}\right\rfloor
+4\right) i_{n+1}} \\
\vdots & \vdots & \vdots & \vdots & \vdots \\
q^{\left( y+2\left\lfloor \frac{n}{2}\right\rfloor +2\right) i_{1}} & \cdots
& q^{\left( y+2\left\lfloor \frac{n}{2}\right\rfloor +2\right) i_{s}} &
\cdots & q^{\left( y+2\left\lfloor \frac{n}{2}\right\rfloor +2\right)
i_{n+1}}%
\end{array}%
\right) $.

Substituting $(x=n,y=-2)$ for even $n$ and $(x=n+1,y=0)$ for odd $n$ into (%
\ref{eq:det2}), respectively, we find that the determinate $A$ is zero.
After a straightforward calculation for the second determinate in (\ref%
{eq:det2}), we obtain the explicit form of $(n+1)$-bracket (\ref{eq:det2})
\begin{align*}
& {\llbracket}L_{i_{1}},L_{i_{2}},\cdots ,L_{i_{n+1}}{\rrbracket}=\frac{%
sign\left( n+1\right) }{\left( q-q^{-1}\right) ^{n}} \\
& \det \left(
\begin{array}{cccc}
q^{-2\left\lfloor \frac{n}{2}\right\rfloor i_{1}} & q^{-2\left\lfloor \frac{n%
}{2}\right\rfloor i_{2}} & \cdots & q^{-2\left\lfloor \frac{n}{2}%
\right\rfloor i_{n+1}} \\
q^{2\left( -\left\lfloor \frac{n}{2}\right\rfloor +1\right) i_{1}} &
q^{2\left( -\left\lfloor \frac{n}{2}\right\rfloor +1\right) i_{2}} & \cdots
& q^{2\left( -\left\lfloor \frac{n}{2}\right\rfloor +1\right) i_{n+1}} \\
\vdots & \vdots & \vdots & \vdots \\
q^{2\left( \left\lfloor \frac{n+1}{2}\right\rfloor -1\right) i_{1}} &
q^{2\left( \left\lfloor \frac{n+1}{2}\right\rfloor -1\right) i_{2}} & \cdots
& q^{2\left( \left\lfloor \frac{n+1}{2}\right\rfloor -1\right) i_{n+1}} \\
q^{2\left\lfloor \frac{n+1}{2}\right\rfloor i_{1}} & q^{2\left\lfloor \frac{%
n+1}{2}\right\rfloor i_{2}} & \cdots & q^{2\left\lfloor \frac{n+1}{2}%
\right\rfloor i_{n+1}}%
\end{array}%
\right) L_{\Sigma _{l=1}^{n+1}i_{l}},
\end{align*}%
which shows that (\ref{eq:qlrbracket}) is satisfied for $n+1$. Now the proof
is completed.\bigskip
\end{proof}

For the $q$-3-bracket (\ref{eq:qv3alg}), we already recognize that it
satisfies the sh-Jacobi's identity (\ref{eq:shuffle}), but the FI (\ref{eq:eFI})
does not hold. Let us consider the case of the $q$-$n$%
-bracket (\ref{eq:qlrbracket}). Taking $A_{i}=L_{-i-1},i=1,2,\cdots ,n-2$, $%
A_{n-1}=L_{\frac{\left( n-1\right) n}{2}}$ and $B_{j}=L_{j-1},j=1,2\cdots ,n$
in (\ref{eq:eFI}), straightforward calculation shows that the left-hand side
of (\ref{eq:eFI}) equals zero, but its right-hand side does not. It
indicates that the FI (\ref{eq:eFI}) does not hold for (\ref%
{eq:qlrbracket}). Therefore the $q$-$n$-bracket relation (\ref{eq:qlrbracket})
is not an $n$-Lie algebra. In spite of this negative result it is
instructive to pursue the analysis of the $q$-$n$-bracket (\ref{eq:qlrbracket}%
).

\begin{proposition}
\label{th:4}When $n\geq 3,$ the $q$-$n$-bracket relation (\ref{eq:qlrbracket})
is a sh-$n$-Lie algebra.
\end{proposition}

\begin{proof}
Let us first focus on (\ref{eq:qlrbracket}) with odd $n$. In terms of the L$%
\acute{e}$vi-Civit$\grave{a}$ symbol (\ref{eq:LV}), we can rewrite $(2n+1)$%
-bracket (\ref{eq:qlrbracket}) as
\begin{equation}
{\llbracket}L_{i_{1}},\cdots ,L_{i_{2n+1}}{\rrbracket}=\frac{sign\left(
2n+1\right) }{\left( q-q^{-1}\right) ^{2n}}\epsilon _{i_{1}\cdots
i_{2n+1}}^{j_{1}\cdots j_{2n+1}}q^{-2nj_{1}+2\left( -n+1\right) j_{2}+\cdots
+2\left( n-1\right) j_{2n}+2nj_{2n+1}}L_{\Sigma _{l=1}^{2n+1}i_{l}}.
\label{eq:qlodd}
\end{equation}%
Then let us use the expression (\ref{eq:qlodd}) to calculate ${\llbracket%
\llbracket}L_{i_{1}},\cdots ,L_{i_{2n+1}}{\rrbracket},L_{i_{2n+2}},\cdots
,L_{i_{4n+1}}{\rrbracket}$. It leads to
\begin{align}
& {\llbracket\llbracket}L_{i_{1}},\cdots ,L_{i_{2n+1}}{\rrbracket}%
,L_{i_{2n+2}},\cdots ,L_{i_{4n+1}}{\rrbracket}  \notag  \label{eq:sht} \\
& =\sum_{k=2}^{2n+2}\frac{\left( -1\right) ^{k}}{\left( q-q^{-1}\right) ^{4n}%
}\epsilon _{i_{1}\cdots i_{2n+1}}^{j_{1}\cdots j_{2n+1}}\epsilon
_{i_{2n+2}\cdots i_{4n+1}}^{j_{2n+2}\cdots \widehat{j_{2n+k}}\cdots
j_{4n+2}}q^{2\left( -2n+k-2\right) j_{1}+2\left( -2n+k-1\right) j_{2}+\cdots
+2\left( k-2\right) j_{2n+1}}  \notag \\
& q^{-2nj_{2n+2}+\cdots +2\left( -n+k-3\right) j_{2n+k-1}+\widehat{2\left(
-n+k-2\right) j_{2n+k}}+2\left( -n+k-1\right) j_{2n+k+1}\cdots
+2nj_{4n+2}}L_{\Sigma _{l=1}^{4n+1}i_{l}}.
\end{align}

Substituting (\ref{eq:sht}) into the left-hand side of $\left( \ref{emisila
sh}\right) $, we obtain
\begin{align}
& \epsilon _{m_{1}\cdots m_{4n+1}}^{i_{1}\cdots i_{4n+1}}{\llbracket%
\llbracket}L_{i_{1}},\cdots ,L_{i_{2n+1}}{\rrbracket},L_{i_{2n+2}},\cdots
,L_{i_{4n+1}}{\rrbracket}  \notag  \label{eq:psht} \\
& =\frac{\left( 2n+1\right) !\left( 2n\right) !}{\left( q-q^{-1}\right) ^{4n}%
}\sum_{k=2}^{2n+2}\left( -1\right) ^{k}\epsilon _{m_{1}\cdots
m_{4n+1}}^{j_{1}\cdots \widehat{j_{2n+k}}\cdots j_{4n+2}}q^{\alpha
}L_{\Sigma _{l=1}^{4n+1}i_{l}},
\end{align}%
where the power of $q$ is given by
\begin{eqnarray}
\alpha &=&2\left( -2n+k-2\right) j_{1}+2\left( -2n+k-1\right) j_{2}+\cdots
+2\left( k-2\right) j_{2n+1}  \notag  \label{eq:alpha} \\
&-2&nj_{2n+2}+\cdots +\widehat{2\left( -n+k-2\right) j_{2n+k}}+\cdots
+2nj_{4n+2},
\end{eqnarray}%
and the following formula is useful in simplifying expression:
\begin{equation}
\epsilon _{m_{1}\cdots m_{n}}^{i_{1}\cdots i_{n}}\epsilon _{i_{1}\cdots
i_{k}}^{j_{1}\cdots j_{k}}=k!\epsilon _{m_{1}\cdots m_{n}}^{j_{1}\cdots
j_{k}i_{k+1}\cdots i_{n}}.\label{equal of emixi}
\end{equation}%
From the expression of $\alpha $ (\ref{eq:alpha}), we observe that the
coefficients of two different $j_{\mu }$ should be equal. Since $\epsilon
_{1\cdots 4n+1}^{j_{1}\cdots \widehat{j_{2n+k}}\cdots j_{4n+2}}$ is
completely antisymmetric, it is easy to see that (\ref{eq:psht}) equals
zero. It indicates that the sh-Jacobi's identity is satisfied by (\ref%
{eq:qlrbracket}) with odd $n$.

For the case of (\ref{eq:qlrbracket}) with even $n$, by the similar way, we
can confirm the corresponding sh-Jacobi's identity. Taking the above
results, we may conclude that the sh-Jacobi's identity (\ref{eq:shJacobi})
does hold for (\ref{eq:qlrbracket}). Since the structure constants are
determined by the the determinant, $n$-bracket (\ref{eq:qlrbracket}) is
anticommutative. Based on the above analysis, it is clear that the
$q$-deformed V-W $n$-algebra is indeed a sh-$n$-Lie algebra.
\end{proof}

We have constructed the $q$-deformed V-W $n$-algebra (\ref{eq:qlrbracket}).
It should be noted that the structure constant of this $q$-deformed
infinite-dimensional $n$-algebra is  determined by the Vandermonde
determinant. In the limit $q\rightarrow 1$, it is easy to see that
(\ref{eq:qlrbracket}) reduces to the null $n$-algebra.
It is also interesting to note
that the structure constant of the $q$-deformed V-W $(n-1)$-algebra
can be induced from that of (\ref{eq:qlrbracket}).
More precisely, it is equal to $\left( -1\right) ^{n-1}\left(q-q^{-1}\right) $
times the structure constant of (\ref{eq:qlrbracket}), where
the original Vandermonde determinant in (\ref{eq:qlrbracket})
is replaced by the $(h\left( n\right) ,n)$-minor of its Vandermonde matrix,
$h\left( n\right) $ takes $1$ and $n$
for odd and even $n$, respectively.

Let us list first few $q$-deformed V-W $n$-algebras as follows:
\begin{eqnarray}
\bullet &&{\llbracket}L_{i_{1}},L_{i_{2}},L_{i_{3}},L_{i_{4}}{\rrbracket}
\notag \\
&=&\left( q-q^{-1}\right) ^{-3}\det \left(
\begin{array}{cccc}
q^{-2i_{1}} & q^{-2i_{2}} & q^{-2i_{3}} & q^{-2i_{4}} \\
1 & 1 & 1 & 1 \\
q^{2i_{1}} & q^{2i_{2}} & q^{2i_{3}} & q^{2i_{4}} \\
q^{4i_{1}} & q^{4i_{2}} & q^{4i_{3}} & q^{4i_{4}}%
\end{array}%
\right) L_{\sum_{k=1}^{4}i_{k}}  \notag \\
&=&\left( q-q^{-1}\right) ^{3}q^{\sum_{k=1}^{4}i_{k}}\prod\limits_{1\leq
m<n\leq 4}[i_{m}-i_{n}]L_{\sum_{k=1}^{4}i_{k}}.
\end{eqnarray}

\begin{eqnarray}
\bullet &&{\llbracket}L_{i_{1}},L_{i_{2}},L_{i_{3}},L_{i_{4}},L_{i_{5}}{%
\rrbracket}  \notag \\
&=&\left( q-q^{-1}\right) ^{-4}\det \left(
\begin{array}{ccccc}
q^{-4i_{1}} & q^{-4i_{2}} & q^{-4i_{3}} & q^{-4i_{4}} & q^{-4i_{5}} \\
q^{-2i_{1}} & q^{-2i_{2}} & q^{-2i_{3}} & q^{-2i_{4}} & q^{-2i_{5}} \\
1 & 1 & 1 & 1 & 1 \\
q^{2i_{1}} & q^{2i_{2}} & q^{2i_{3}} & q^{2i_{4}} & q^{2i_{5}} \\
q^{4i_{1}} & q^{4i_{2}} & q^{4i_{3}} & q^{4i_{4}} & q^{4i_{5}}%
\end{array}%
\right) L_{\sum_{k=1}^{5}i_{k}}  \notag \\
&=&\left( q-q^{-1}\right) ^{6}q^{\sum_{k=1}^{5}i_{k}}\prod\limits_{1\leq
m<n\leq 5}[i_{m}-i_{n}]L_{\sum_{k=1}^{5}i_{k}}.
\end{eqnarray}
\begin{eqnarray}
\bullet &&{\llbracket}%
L_{i_{1}},L_{i_{2}},L_{i_{3}},L_{i_{4}},L_{i_{5}},L_{i_{6}}{\rrbracket}
\notag \\
&=&-\left( q-q^{-1}\right) ^{-5}\det \left(
\begin{array}{cccccc}
q^{-4i_{1}} & q^{-4i_{2}} & q^{-4i_{3}} & q^{-4i_{4}} & q^{-4i_{5}} &
q^{-4i_{6}} \\
q^{-2i_{1}} & q^{-2i_{2}} & q^{-2i_{3}} & q^{-2i_{4}} & q^{-2i_{5}} &
q^{-2i_{6}} \\
1 & 1 & 1 & 1 & 1 & 1 \\
q^{2i_{1}} & q^{2i_{2}} & q^{2i_{3}} & q^{2i_{4}} & q^{2i_{5}} & q^{2i_{6}}
\\
q^{4i_{1}} & q^{4i_{2}} & q^{4i_{3}} & q^{4i_{4}} & q^{4i_{5}} & q^{4i_{6}}
\\
q^{6i_{1}} & q^{6i_{2}} & q^{6i_{3}} & q^{6i_{4}} & q^{6i_{5}} & q^{6i_{6}}%
\end{array}%
\right) L_{\sum_{k=1}^{6}i_{k}}  \notag \\
&=&-\left( q-q^{-1}\right) ^{10}q^{\sum_{k=1}^{6}i_{k}}\prod\limits_{1\leq
m<n\leq 6}[i_{m}-i_{n}]L_{\sum_{k=1}^{6}i_{k}}.
\end{eqnarray}


\section{$q$-deformed $SDiff(T^2)$ $n$-algebra}



\subsection{Sine $3$-algebra and $q$-deformed $SDiff(T^2)$ $3$-algebra}


The quantum differential calculus on the quantum plane $\mathbf{C}_{q}[x,y]$
have been well investigated \cite{JMP92}.
For the the quantum plane $\mathbf{C}_{q}[x,y]$, each of its elements is a
finite linear combination of the monomes $y^{n}x^{m}$, satisfying
\begin{equation}
x^{m}y^{n}=q^{nm}y^{n}x^{m}, \ \ m,n\in \mathbf{N}.  \label{1}
\end{equation}

The Gauss derivatives on $\mathbf{C}_{q}[x,y]$ can be extended to be formal
pseud-differential operators $D_{x},D_{y}$ which can be defined on the set $%
\mathbf{C}_{q}[[x,y]]$ and satisfy%
\begin{equation}
D_{x}^{n}D_{y}^{m}=q^{nm}D_{y}^{m}D_{x}^{n},\ \ m,n\in \mathbf{Z},
\label{dq2}
\end{equation}%
where $\mathbf{C}_{q}[[x,y]]$ is a set of all Laurent series in $x,y$ such
that $\left( \ref{1}\right) $ is valid for $n,m\in \mathbf{Z.}$

In terms of the pseud-differntial operators $D_{x}$ and $D_{y}$,
Kinani et al. \cite{Kinani357} introduced the following generators:
\begin{eqnarray}  \label{qgt}
T_{n}=q^{n_{1}\cdot n_{2}/2}\cdot D_{y}^{n_{1}}D_{x}^{n_{2}},
\end{eqnarray}
where $n=(n_{1},n_{2})\in \mathbf{Z}^2$. In the rest of this paper, we
denote the subscript $l$ on $T_l$ being a two-dimensional vector with
integer components.

By means of $\left( \ref{dq2}\right) $, it is easy to verify that the
generators $T_{n}$ (\ref{qgt}) satisfy
\begin{equation}
T_{n}T_{m}=q^{\frac{1}{2}m\wedge n}T_{n+m},  \label{pro of T}
\end{equation}%
where $m\wedge n=m_{1}n_{2}-m_{2}n_{1}$.

Thus we have the algebra
\begin{equation}
[T_{m}, T_{n}]=T_{m}T_{n}-T_{n}T_{m}=(q^{\frac{1}{2}%
n\wedge m}-q^{\frac{1}{2}m\wedge n})T_{m+n}.  \label{tt1}
\end{equation}

When $q=exp(-2\pi $\textbf{$i$}$\alpha )$, (\ref{tt1}) becomes the
sine algebra \cite{Fairlie224}
\begin{equation}\label{sine1}
[T_{m}, T_{n}]=2\mathbf{i}\sin (\pi \alpha m\wedge n)T_{m+n},
\end{equation}
where $\alpha$ is an arbitrary constant.

Taking the rescaled generators $\bar{T}_{m}=-\frac{\mathbf{i}}{2\pi
\alpha }T_{m}$, we note that in the limit $\alpha\rightarrow 0$,
(\ref{sine1}) leads to the $SDiff(T^{2})$ algebra \cite{Floratos201,Antoniadis300}
\begin{equation}
\lbrack \bar{T}_{m}, \bar{T}_{n}]=(m\wedge n)\bar{T}_{m+n}.  \label{sdiff}
\end{equation}

The $q$-deformation of $SDiff(T^{2})$ algebra (\ref{sdiff}) is given by \cite{Kinani357}
\begin{equation}
\lbrack \bar{T}_{m}, \bar{T}_{n}]_{(q^{\frac{3}{2}m\wedge n},q^{\frac{3}{2}%
n\wedge m})}=q^{\frac{3}{2}m\wedge n}\bar{T}_{m} \bar{T}_{n}-q^{\frac{3%
}{2}n\wedge m}\bar{T}_{n} \bar{T}_{m}=[m\wedge n]\bar{T}_{m+n},
\label{1.2}
\end{equation}%
where $\bar{T}_{m}=\frac{1}{q-q^{-1}}T_{m}$.

Let us turn to the case of 3-algebra. Substituting the generators
(\ref{qgt}) into the operator Nambu 3-bracket (\ref{eq:3bracket}) and using
(\ref{pro of T}) and (\ref{tt1}), by direct calculation, we may derive the following 3-algebra:
\begin{eqnarray}\label{qq4}
[T_{m}, T_{n}, T_{k}]
&=&(-q^{\frac{1}{2}\left( m\wedge n-n\wedge k+k\wedge m\right) }+q^{-\frac{1%
}{2}\left( m\wedge n-n\wedge k+k\wedge m\right) } \notag \\
&-&q^{\frac{1}{2}\left(
m\wedge n+n\wedge k-k\wedge m\right) }+q^{-\frac{1}{2}\left( m\wedge
n+n\wedge k-k\wedge m\right) }  \notag \\
&-&q^{\frac{1}{2}\left( -m\wedge n+n\wedge k+k\wedge m\right) }+q^{-\frac{1}{%
2}\left( -m\wedge n+n\wedge k+k\wedge m\right) })T_{m+n+k}.
\end{eqnarray}
Performing straightforward calculations, we find that the
3-algebra with the $q$ parameter (\ref{qq4}) does not satisfy the
FI (\ref{eq:FI}) and the sh-Jacobi's identity (\ref{eq:shuffle}).

An interesting case is for the special value of $q$. Taking $q=e^{-\pi
\mathbf{i}}$, we may rewrite (\ref{qq4}) as
\begin{eqnarray}
[T_{m}, T_{n}, T_{k}]&=&2\mathbf{i}(\sin (\frac{\pi }{2}\left( m\wedge
n-n\wedge k+k\wedge m\right) )+\sin (\frac{\pi }{2}\left( m\wedge n+n\wedge
k-k\wedge m\right) )  \notag  \label{3 of sin} \\
&+&\sin (\frac{\pi }{2}\left( -m\wedge n+n\wedge k+k\wedge m\right)
))T_{m+n+k}.
\end{eqnarray}

Not as the case of (\ref{qq4}), an intriguing property of (\ref{3 of sin})
is that it does satisfy the FI (\ref{eq:FI}). Since the skew symmetry  also
holds, the sine 3-algebra (\ref{3 of sin}) is indeed a Fillipov 3-algebra.

Let us take the rescaled generators $\bar{T}_{n}=\frac{1}{\left(
q-q^{-1}\right) ^{1/2}}T_{n}$ and define the $q$-3-bracket
\begin{eqnarray}
{\llbracket}\bar{T}_{m}, \bar{T}_{n}, \bar{T}_{k}{\rrbracket} &=&\bar{T}%
_{m}\ast \lbrack \bar{T}_{n},\bar{T}_{k}]_{(q^{\frac{3}{2}n\wedge k},q^{%
\frac{3}{2}k\wedge n})}+\bar{T}_{n}\ast \lbrack \bar{T}_{k},\bar{T}%
_{m}]_{(q^{\frac{3}{2}k\wedge m},q^{\frac{3}{2}m\wedge k})}  \notag
\label{eq:Tq3bracket} \\
&+&\bar{T}_{k}\ast \lbrack \bar{T}_{m},\bar{T}_{n}]_{(q^{\frac{3}{2}m\wedge
n},q^{\frac{3}{2}n\wedge m})},
\end{eqnarray}%
where the star product is given by
\begin{equation}
\bar{T}_{m}\ast \lbrack \bar{T}_{n}, \bar{T}_{k}]_{(q^{\frac{3}{2}n\wedge
k},q^{\frac{3}{2}k\wedge n})}=q^{\frac{3}{2}m\wedge \left( n+k\right) }\bar{T%
}_{m}\lbrack \bar{T}_{n}, \bar{T}_{k}]_{(q^{\frac{3}{2}n\wedge k},q^{%
\frac{3}{2}k\wedge n})}.  \label{star pro of q-torus}
\end{equation}%
Then we have the $q$-deformed $SDiff(T^{2})$ 3-algebra
\begin{eqnarray}\label{qsdiff3}
{\llbracket}\bar{T}_{m}, \bar{T}_{n},\bar{T}_{k}{\rrbracket}
&=&([m\wedge n-n\wedge k+k\wedge m]+[m\wedge n+n\wedge k-k\wedge
m]\notag\\
&+&[-m\wedge n+n\wedge k+k\wedge m])\bar{T}_{m+n+k}  \notag \\
&=&([\det \left(
\begin{array}{ccc}
m_{1} & n_{1} & k_{1} \\
m_{2} & n_{2} & k_{2} \\
-1 & 1 & 1%
\end{array}%
\right) ]+[\det \left(
\begin{array}{ccc}
m_{1} & n_{1} & k_{1} \\
m_{2} & n_{2} & k_{2} \\
1 & -1 & 1%
\end{array}%
\right) ]\notag\\
&+&[\det \left(
\begin{array}{ccc}
m_{1} & n_{1} & k_{1} \\
m_{2} & n_{2} & k_{2} \\
1 & 1 & -1%
\end{array}%
\right) ])\bar{T}_{m+n+k}.
\end{eqnarray}%

As the case of (\ref{qq4}),  the infinite-dimensional $q$-deformed
3-algebra (\ref{qsdiff3}) does not satisfy the FI (\ref{eq:FI})
and the sh-Jacobi's identity (\ref{eq:shuffle}).

In the limit $q\rightarrow 1$, (\ref{qsdiff3}) reduces to the $SDiff(T^2)$
3-algebra \cite{Chen2012}
\begin{eqnarray}\label{47}
[\bar{T}_{m}, \bar{T}_{n}, \bar{T}_{k}]=\left( m\wedge n+n\wedge k+k\wedge
m\right) \bar{T}_{m+n+k}.
\end{eqnarray}
The FI (\ref{eq:FI}) is satisfied for this infinite-dimensional
3-algebra.
Taking $\bar{T}_{k}=\bar{T}_{0}$ in (\ref{47}),
(\ref{47}) can be regarded as the parametrized bracket relation
$[\bar{T}_{m}, \bar{T}_{n}]_{\bar{T}_{0}}$.
This parametrized bracket relation gives rise to the $SDiff(T^2)$
algebra (\ref{sdiff}).


\subsection{(co)Sine $n$-algebra}


For the generators $T_{n}$ (\ref{qgt}), we note that they are the
associative operators with the product (\ref{pro of T}). According to the
definition of the $n$-bracket (\ref{eq:lbracket}), we get the following
result.

\begin{theorem}
\label{th:4 copy(1)}The generators (\ref{qgt}) satisfy
the following closed algebraic structure relation:%
\begin{equation}
{[}T_{i_{1}}, \cdots , T_{i_{n}}{]}=\epsilon _{i_{1}\cdots i_{n}}^{j_{1}\cdots
j_{n}}q^{\frac{1}{2}\Sigma _{k>s}j_{k}\wedge j_{s}}T_{\Sigma
_{l=1}^{n}i_{l}}.  \label{prod of T}
\end{equation}
\end{theorem}

\begin{proof}
The $n$-bracket $\left( \ref{prod of T}\right) $ \ will follow from $\left( \ref%
{pro of T}\right) $ if we can show that
\begin{equation}
{[}T_{i_{1}}, \cdots , T_{i_{n}}{]}=\epsilon _{i_{1}\cdots i_{n}}^{j_{1}\cdots
j_{n}}T_{j_{1}} T_{j_{1}} \cdots T_{j_{n}}.
\label{commu of T copy(1)}
\end{equation}%
First, let us prove $\left( \ref{commu of T copy(1)}\right) $ by the
mathematical induction for $n$. By $\left( \ref{tt1}\right) ,$ it is obvious
that $\left( \ref{commu of T copy(1)}\right) $ holds for $n=2$.
We suppose $\left( \ref{commu of T copy(1)}\right) $ is satisfied for $n$-bracket.
Note that the generators $T_{i}$ are the associative operators
under the product (\ref{pro of T}), we obtain
\begin{eqnarray}
\lbrack T_{i_{1}}, \cdots , T_{i_{n+1}}] &=&\sum_{l=1}^{n+1}\left( -1\right)
^{l-1}T_{i_{l}} \lbrack T_{i_{1}},\cdots ,\hat{T}_{i_{l}},\cdots
,T_{i_{n+1}}] \notag \\
&=&\sum_{l=1}^{n+1}\left( -1\right) ^{l-1}\epsilon _{i_{1}\cdots \hat{i}%
_{l}\cdots i_{n+1}}^{j_{2}\cdots j_{n+1}}T_{i_{l}} \left(
T_{j_{2}} \cdots  T_{j_{n+1}}\right) \notag \\
&=&\sum_{l=1}^{n+1}\left( -1\right) ^{l-1}\epsilon _{i_{1}\cdots \hat{i}%
_{l}\cdots i_{n+1}}^{j_{2}\cdots j_{n+1}}\left( \delta
_{i_{l}}^{j_{1}}T_{j_{1}}\right)  \left( T_{j_{2}} \cdots
T_{j_{n+1}}\right) \notag \\
&=&\left( \sum_{l=1}^{n+1}\left( -1\right) ^{l-1}\epsilon _{i_{1}\cdots \hat{%
i}_{l}\cdots i_{n+1}}^{j_{2}\cdots j_{n+1}}\delta _{i_{l}}^{j_{1}}\right)
T_{j_{1}} T_{j_{2}} \cdots  T_{j_{n+1}} \notag \\
&=&\epsilon _{i_{1}\cdots i_{n+1}}^{j_{1}\cdots j_{n+1}}T_{j_{1}}
T_{j_{2}} \cdots T_{j_{n+1}},
\end{eqnarray}
which shows that $\left( \ref{commu of T copy(1)}\right) $ is also
satisfied for $n+1$-bracket.

Substituting (\ref{pro of T}) into (\ref{commu of T copy(1)}),  we obtain ( \ref{prod of T}).
The proof is completed.
\end{proof}

When $n=3$ in ( \ref{prod of T}), we have known that the corresponding 3-algebra
(\ref{qq4}) does not satisfy the FI (\ref{eq:FI}) and the sh-Jacobi's
identity (\ref{eq:shuffle}).
Let us now analyze the property of $n$-algebra (\ref{prod of T}) for $n\ge 4$.

\begin{proposition}
\label{Proposition n even is sh}When $n$ is even, the $n$-algebra (\ref{prod
of T}) is a sh-$n$-Lie algebra.
\end{proposition}

\begin{proof}
Due to the skew-symmetry of $\epsilon _{i_{1}\cdots i_{n}}^{j_{1}\cdots
j_{n}}$ in (\ref{prod of T}), it is obvious that the skew-symmetry holds for
the $n$-bracket (\ref{prod of T}). We are left to show that the sh-Jacobi's
identity
\begin{eqnarray}\label{sjtn}
\epsilon _{k_{1}\cdots k_{2n-1}}^{i_{1}\cdots i_{2n-1}}[\left[
T_{i_{1}},\cdots ,T_{i_{n}}\right] ,T_{i_{n+1}},\cdots ,T_{i_{2n-1}}]=0
\end{eqnarray}
is satisfied.

Substituting (\ref{prod of T}) \ and (\ref{commu of T copy(1)})
into the left-hand side of the sh-Jacobi's identity (\ref{sjtn}), we obtain
\begin{eqnarray}\label{sjip}
&&\epsilon _{k_{1}\cdots k_{2n-1}}^{i_{1}\cdots i_{2n-1}}\epsilon
_{i_{1}\cdots i_{n}}^{j_{1}\cdots j_{n}}q^{\frac{1}{2}\Sigma _{1\leq s<k\leq
n}j_{k}\wedge j_{s}}[T_{\Sigma _{m=1}^{n}i_{m}},T_{i_{n+1}},\cdots
,T_{i_{2n-1}}]  \notag \\
&=&\sum_{l=0}^{n-1}\left( -1\right) ^{l}\epsilon _{k_{1}\cdots
k_{2n-1}}^{i_{1}\cdots i_{2n-1}}\epsilon _{i_{1}\cdots i_{n}}^{j_{1}\cdots
j_{n}}\epsilon _{\Sigma _{k=1}^{n}i_{k},i_{n+1},\cdots
,i_{2n-1}}^{j_{n+1}\cdots j_{n+l},\Sigma _{m=1}^{n}i_{m},j_{n+l+2},\cdots
,j_{2n}}  \notag \\
&& q^{\frac{1}{2}\Sigma _{1\leq s<k\leq n}j_{k}\wedge
j_{s}}T_{j_{n+1}} \cdots \cdot T_{j_{n+l}} \left( T_{\Sigma
_{k=1}^{n}i_{k}}\right)  T_{j_{n+l+2}} \cdots  T_{j_{2n-1}} \notag \\
&=&\sum_{l=0}^{n-1}\left( -1\right) ^{l}\epsilon _{k_{1}\cdots
k_{2n-1}}^{i_{1}\cdots i_{2n-1}}\epsilon _{i_{1}\cdots i_{n}}^{j_{1}\cdots
j_{n}}\epsilon _{i_{n+1}\cdots i_{2n-1}}^{j_{n+1}\cdots
j_{2n-1}}T_{j_{n+1}} \cdots  T_{j_{n+l}} \left(
T_{j_{1}} \cdots  T_{j_{n}}\right)  T_{j_{n+l+1}} \cdots
 T_{j_{2n-1}}  \notag \\
&=&\sum_{l=0}^{n-1}\left( -1\right) ^{l+l n}n!\left( n-1\right)
!\epsilon _{k_{1}\cdots k_{2n-1}}^{j_{1}\cdots j_{2n-1}}T_{j_{1}}
\cdots  T_{j_{2n-1}}  \notag \\
&=&n!\left( n-1\right) !\epsilon _{k_{1}\cdots k_{2n-1}}^{j_{1}\cdots
j_{2n-1}}q^{\frac{1}{2}\Sigma _{1\leq s<k\leq 2n-1}j_{k}\wedge
j_{s}}\sum_{l=0}^{n-1}\left( -1\right) ^{l\left( n+1\right) }T_{\Sigma
_{m=1}^{2n-1}k_{m}}.
\end{eqnarray}
Since $\sum_{l=0}^{n-1}\left( -1\right) ^{l\left(
n+1\right) }=0$ with even $n$, the right-hand side of (\ref{sjip})
equals zero. Therefore the sh-Jacobi's identity (\ref{sjtn}) holds for even $n$.
The proof is completed.
\end{proof}

We finally remark that when $n$ is odd, the $n$-algebra (\ref{prod of T})
is not a sh-$n$-Lie algebra.
For (\ref{sjip}) with odd $n$,  the coefficient of $T_{\Sigma _{m=1}^{2n-1}k_{m}}$
is
\begin{equation}
\left( n!\right) ^{2}\epsilon _{k_{1}\cdots k_{2n-1}}^{j_{1}\cdots
j_{2n-1}}q^{\frac{1}{2}\Sigma _{1\leq s<k\leq 2n-1}j_{k}\wedge j_{s}}.
\end{equation}%
Let us choose $k_{l}=\left( l,1\right)$,
we note that the coefficient of the monomial with the maximal power is
\begin{equation}
\left( n!\right) ^{2}\epsilon _{k_{1}\cdots k_{2n-1}}^{k_{1}\cdots
k_{2n-1}}=\left( n!\right) ^{2},
\end{equation}%
It is obvious that the sh-Jacobi's identity does not hold for this case.

By the similar way, we can confirm that the $n$-algebra (\ref{prod of T})
does not also satisfy the FI (\ref{eq:eFI}).

We have derived the $n$-bracket (\ref{prod of T}) with general $q\in \mathbf{%
C}$. Let us now focus on the case of the special value of $q$. Taking $%
q=\exp \left( -2\pi \mathbf{i}\alpha \right) $,  $\alpha\in R $,
we can express the $n$-bracket (\ref{prod of T}) as%
\begin{eqnarray}
&&[T_{i_{1}}, \cdots , T_{i_{n}}]  \notag  \label{sinnalg} \\
&=&\frac{1}{2}\left( \epsilon _{i_{1}\cdots i_{n}}^{j_{1}\cdots j_{n}}\exp
\left( \pi \mathbf{i}\alpha \Sigma _{k<s}j_{k}\wedge j_{s}\right) +\epsilon
_{i_{1}\cdots i_{n}}^{j_{n}\cdots j_{1}}\exp \left( \pi \mathbf{i}\alpha
\Sigma _{k>s}j_{k}\wedge j_{s}\right) \right) T_{\Sigma _{l=1}^{n}i_{l}}
\notag \\
&=&\frac{1}{2}\epsilon _{i_{1}\cdots i_{n}}^{j_{1}\cdots j_{n}}\cos \left(
\pi \alpha \Sigma _{k<s}j_{k}\wedge j_{s}\right) \left( 1+\left( -1\right) ^{%
\frac{n\left( n-1\right) }{2}}\right) T_{\Sigma _{l=1}^{n}i_{l}}  \notag \\
&&+\frac{\mathbf{i}}{2}\epsilon _{i_{1}\cdots i_{n}}^{j_{1}\cdots j_{n}}\sin
\left( \pi \alpha \Sigma _{k<s}j_{k}\wedge j_{s}\right) \left( 1-\left(
-1\right) ^{\frac{n\left( n-1\right) }{2}}\right) T_{\Sigma _{l=1}^{n}i_{l}}.
\end{eqnarray}

When $n$ is even, (\ref{sinnalg}) is a sh-$n$-Lie algebra. However not as
the case of (\ref{prod of T}) with odd $n$, we find that when $n=3$,
for the special value  $\alpha =\frac{1}{2}$,
(\ref{sinnalg}) gives a Fillipov 3-algebra (\ref{3 of sin}).

Let us consider the case of $n=5$. In this case, (\ref{sinnalg}) gives
\begin{equation}
\lbrack T_{i_{1}}, \cdots , T_{i_{5}}]=\epsilon _{i_{1}\cdots
i_{5}}^{j_{1}\cdots j_{5}}\cos \left( \pi \alpha \Sigma _{k<l}j_{k}\wedge
j_{l}\right) T_{i_{1}+\cdots +i_{5}}.  \label{sin5alg}
\end{equation}%
Taking $\alpha =\frac{1}{3}$ in (\ref{sin5alg}), it is interesting to
note that the sh-Jacobi's identity (\ref{eq:shJacobi}) holds,
but the FI (\ref{eq:eFI}) fails in this example.
Thus for this special $\alpha $, the cosine 5-algebra (\ref{sin5alg})
gives a sh-5-Lie algebra.


\subsection{$q$-deformed $SDiff(T^2)$ $n$-algebra}


To construct the $q$-deformed $SDiff(T^2)$ $n$-algebra, let us define a $q$-$%
n$-bracket as follows:
\begin{equation}
{\llbracket}\bar T_{i_{1}}, \bar T_{i_{2}},\cdots , \bar T_{i_{n}}{\rrbracket%
}=\sum\limits_{s=1}^{n}\left( -1\right) ^{s+1}\bar T_{i_{s}}\ast {\llbracket}%
\bar T_{i_{1}}, \bar T_{i_{2}},\cdots ,\hat{\bar T}_{i_{s}},\cdots , \bar
T_{i_{n}}{\rrbracket},  \label{eq:lbracket1}
\end{equation}%
where the generators are $\bar{T}_{n}=\frac{1}{\left( q-q^{-1}\right)
^{1/\left( n-1\right) }}T_{n}$ and the general star product is given by
\begin{equation}
\bar T_{i_{1}}\ast {\llbracket}\bar T_{i_{2}}, \bar T_{i_{3}},\cdots , \bar
T_{i_{n}}{\rrbracket}=q^{\frac{3}{2}i_{1}\wedge \left( i_{2}+\cdots
+i_{n}\right) }\bar T_{i_{1}} {\llbracket}\bar T_{i_{2}}, \bar
T_{i_{3}},\cdots , \bar T_{i_{n}}{\rrbracket},
\end{equation}%
According to the definition of the $n$-bracket (\ref{eq:lbracket1}),
in similarity with the case of (\ref{prod of T}),
we may derive the following $q$-deformed $SDiff(T^2)$ $n$-algebra:
\begin{equation}
{\llbracket}\bar{T}_{i_{1}}, \cdots , \bar{T}_{i_{n}}{\rrbracket}=\frac{%
\epsilon _{i_{1}\cdots i_{n}}^{j_{1}\cdots j_{n}}q^{\Sigma _{k<s}j_{k}\wedge
j_{s}}}{q-q^{-1}}\bar{T}_{\Sigma _{l=1}^{n}i_{l}}.  \label{ttt'}
\end{equation}
When $n=3$, (\ref{ttt'}) gives  the $q$-deformed $SDiff(T^2)$ 3-algebra (\ref{qsdiff3}).
In the limit $q\rightarrow 1$,  it is not hard to verify that (\ref{ttt'})
reduces to the null $n$-algebra
for $n\geq 4$,
\begin{equation} \label{nnalg}
[\bar{T}_{i_{1}},\cdots ,\bar{T}_{i_{n}}]=0.
\end{equation}

\begin{proposition}
\label{Proposition n even is sh copy(2)}When $n$ is even, the $n$-algebra (%
\ref{ttt'}) is a sh-$n$-Lie algebra.
\end{proposition}

\begin{proof}
Due to the skew-symmetry of $\epsilon _{i_{1}\cdots i_{n}}^{j_{1}\cdots
j_{n}}$ in (\ref{ttt'}), it is obvious that the skew-symmetry holds for the $%
n$-bracket (\ref{ttt'}). We are left to show that the sh-Jacobi's identity
is satisfied.

Substituting (\ref{ttt'}) into the left-hand side of (\ref{sjtn})
and using the formula (\ref{equal of emixi}), we get
\begin{eqnarray}\label{311}
&&\frac{\epsilon _{k_{1}\cdots k_{2n-1}}^{i_{1}\cdots i_{2n-1}}\epsilon
_{i_{1}\cdots i_{n}}^{j_{1}\cdots j_{n}}q^{\Sigma _{1\leq k<l\leq
n}j_{k}\wedge j_{l}}}{q-q^{-1}}{\llbracket}\bar{T}_{\Sigma _{k=1}^{n}j_{k}},%
\bar{T}_{i_{n+1}},\cdots ,\bar{T}_{i_{2n-1}}{\rrbracket}  \notag \\
&=&\frac{1}{\left( q-q^{-1}\right) ^{2}}\sum_{s=0}^{n-1}\left( -1\right)
^{s}\epsilon _{k_{1}\cdots k_{2n-1}}^{i_{1}\cdots i_{2n-1}}\epsilon
_{i_{1}\cdots i_{n}}^{j_{1}\cdots j_{n}}\epsilon _{i_{n+1}\cdots
i_{2n-1}}^{j_{n+1}\cdots j_{2n-1}}q^{\Sigma _{1\leq k<l\leq n}\left(
j_{k}\wedge j_{l}\right) }q^{\Sigma _{n+1\leq k<l\leq 2n-1}\left(
j_{k}\wedge j_{_{l}}\right) }  \notag \\
&&\cdot q^{\Sigma _{n+1\leq l\leq n+s}j_{l}\wedge \left( \Sigma
_{k=1}^{n}j_{k}\right) }q^{\Sigma _{n+s+1\leq l\leq 2n-1}\left( \Sigma
_{k=1}^{n}j_{k}\right) \wedge j_{l}}\bar{T}_{\Sigma _{k=1}^{2n-1}i_{k}} \notag \\
&=&\frac{1}{\left( q-q^{-1}\right) ^{2}}\sum_{s=0}^{n-1}\left( -1\right)
^{s}n!\left( n-1\right) !\epsilon _{k_{1}\cdots k_{2n-1}}^{j_{1}\cdots
j_{2n-1}}q^{\Sigma _{1\leq k<l\leq n}\left( j_{k}\wedge j_{l}\right)
}q^{\Sigma _{n+1\leq k<l\leq 2n-1}
\left( j_{k}\wedge j_{_{l}}\right) }\notag \\
&&\cdot q^{\Sigma _{n+1\leq l\leq n+s}j_{l}\wedge \left( \Sigma
_{k=1}^{n}j_{k}\right) }q^{\Sigma _{n+s+1\leq l\leq 2n-1}\left( \Sigma
_{k=1}^{n}j_{k}\right) \wedge j_{l}}\bar{T}_{\Sigma _{k=1}^{2n-1}i_{k}}.
\end{eqnarray}%

Let us change the indices $\left( j_{1},\cdots ,j_{2n-1}\right) $ to be $\left(
j_{s+1},\cdots ,j_{s+n},j_{1},\cdots j_{s},j_{n+s+1},\cdots ,j_{2n-1}\right)
,$
thus the right-hand side of (\ref{311}) can be rewritten as
\begin{eqnarray}\label{312}
&&\frac{1}{\left( q-q^{-1}\right) ^{2}}\sum_{s=0}^{n-1}\left( -1\right)
^{s}n!\left( n-1\right) !\epsilon _{k_{1}\cdots k_{2n-1}}^{j_{s+1}\cdots
j_{s+n}j_{1}\cdots j_{s}j_{n+s+1}\cdots j_{2n-1}}q^{\Sigma _{1\leq k<l\leq
2n-1}j_{k}\wedge j_{l}}T_{\Sigma _{k=1}^{2n-1}i_{k}}  \notag \\
&=&\frac{1}{\left( q-q^{-1}\right) ^{2}}n!\left( n-1\right) !\epsilon
_{k_{1}\cdots k_{2n-1}}^{j_{1}\cdots j_{2n-1}}q^{\Sigma _{1\leq k<l\leq
2n-1}j_{k}\wedge j_{l}}\sum_{s=0}^{n-1}\left( -1\right) ^{s\left( n+1\right)
}T_{\Sigma _{k=1}^{2n-1}i_{k}}.
\end{eqnarray}%
When $n$ is even, we have $\sum_{s=0}^{n-1}\left( -1\right) ^{s\left(
n+1\right) }=0$.  It indicates that (\ref{311}) equals zero.
Therefore the sh-Jacobi's identity holds.
The proof is completed.
\end{proof}

As the case of (\ref{prod of T}), it is easy to check that when $n$ is odd,
the $n$-algebra (\ref{ttt'}) does not satisfy the sh-Jacobi's identity.
That is to say that (\ref{ttt'}) with odd $n$ is not a sh-$n$-Lie algebra.


\section{A physical realization of the (co)sine $n$-algebra}

Let consider a spinless non-relativistic electron moving on a Bravais
lattice in the $xy$-plane under the influence of a constant uniform magnetic
field $\mathbf{B}=B\mathbf{e}_{z}$. The Hamiltonian is \cite{Dereli}
\begin{equation}
H=\frac{1}{2\mu }(\pi _{x}^{2}+\pi _{y}^{2})+V\left( x,y\right) ,
\end{equation}
where the substrate potential $V(x,y)$ is periodic in $x$ and $y$, i.e., $%
V(x+a_{1},y)=V(x,y+a_{2})=V(x,y)$, with $a_{1}$ and $a_{2}$ being the unit
lattice spacing, the kinetic momentum operators are define by
\begin{equation}
\pi _{x}=p_{x}-\frac{e}{c}A_{x},\ \ \ \ \pi _{y}=p_{y}-\frac{e}{c}A_{y},
\end{equation}
in which $p_{x}=-\mathbf{i}\hbar \frac{\partial }{\partial x}$ and $p_{y}=-%
\mathbf{i}\hbar \frac{\partial }{\partial y}$ are the canonical momentum
operators, $\mathbf{A}=(A_{x},A_{y})$ is the vector potential and can be
given by
\begin{equation}
A_{x}=-\frac{B}{2}y+\frac{\partial \Lambda }{\partial x},\ \ \ \ A_{y}=-%
\frac{B}{2}x+\frac{\partial \Lambda }{\partial y}.
\end{equation}
Here $\Lambda $ is an arbitrary scalar function determining the gauge. For
simplicity, we choose $\Lambda =\frac{1}{2}Bxy$.

We take an arbitrary Bravais lattice vector as follows:
\begin{eqnarray}
{\bm R}_{m}=m_1{\bm a}_1+m_2{\bm a}_2,
\end{eqnarray}
where $m=(m_1, m_2)\in \mathbf{Z}^2$, ${\bm a}_1$ and ${\bm a}_2$ are two
given vectors in the directions ${\bm e}_x$ and ${\bm e}_y$, respectively.

Let ${\bm\beta} =\left( \beta _{1},\beta _{2}\right) $ be a vector which is
classically connected with the cyclotron center given by
\begin{eqnarray}
\beta _{1}=\pi _{x}-\mu \omega y,\ \ \beta _{2}=\pi _{y}+\mu \omega x,
\end{eqnarray}
where $\omega=eB/\mu c$ is the Larmor frequency.

Let us take the magnetic translation operators \cite{Dereli}
\begin{eqnarray}  \label{PGTS}
T_{m}=\exp (\sqrt{2\pi}\mathbf{i}{\bm R}_{m}\cdot {\bm \beta} /\hbar).
\end{eqnarray}
Note that they satisfy
\begin{eqnarray}  \label{TTTP}
T_{m}T_{n}=T_{m+n}\exp \left( \pi \mathbf{i}\alpha m\wedge n\right) ,
\end{eqnarray}
where $\alpha =\phi _{1}/\phi _{0}$ is the number of fluxons passing through
the unit cell, in which $\phi _{1}=\left( a_{1}\wedge a_{2}\right) \cdot
B,\phi _{0}=hc/e$ are the magnetic flux through the unit cell $a_{1}\wedge
a_{2},$ and the magnetic flux quantum, respectively.

For the magnetic translation operators (\ref{PGTS}), it was found that they generate the
the infinite-dimensional sine algebra (\ref{sine1}) \cite{Dereli}.
For the rescaled generators $\bar{T}_{m}=-\frac{\mathbf{i}}{2\pi \alpha }T_{m}$,
the $SDiff(T^{2})$ algebra (\ref{sdiff}) is recovered
in the limit $\alpha \rightarrow 0$.

Since the product of the generators satisfies (\ref{TTTP}), according to
theorem 5, it is known that the $n$ algebra with respect to the generators (%
\ref{PGTS}) is (\ref{sinnalg}).
Thus in terms of the magnetic translation operators (\ref{PGTS}),
we give an explicit physical realization of the (co)sine $n$-algebra (\ref{sinnalg}).

Let us now discuss first few $n$-algebras.
\begin{eqnarray}\label{Psin3}
\bullet \ \ [T_{i_{1}}, T_{i_{2}}, T_{i_{3}}] &=&2\mathbf{i}\left[ \sin \left(
\pi \alpha \left( {i}_{1}\wedge i_{2}+i_{1}\wedge {i}_{3}+i_{2}\wedge
i_{3}\right) \right) -\sin (\pi \alpha \left( -{i}_{1}\wedge
i_{2}+i_{1}\wedge {i}_{3}+i_{2}\wedge i_{3}\right) \right.  \notag\\
&&\left. -\sin \left( \pi \alpha \left( {i}_{1}\wedge i_{2}+i_{1}\wedge {i}%
_{3}-i_{2}\wedge {i}_{3}\right) \right) \right] T_{i_{1}+i_{2}+i_{3}}
\end{eqnarray}

Comparing (\ref{Psin3}) with (\ref{3 of sin}), we see that when $\alpha =%
\frac{1}{2}$, (\ref{Psin3}) becomes a Fillipov 3-algebra.
Let us take the rescaled generators 
$\bar{T}_{i_{j}}=\sqrt{\frac{-\mathbf{i}}{2\pi \alpha }}T_{i_{j}}$, $j=1,2,3$,
we see that in the limit $\alpha\rightarrow 0$, (\ref{Psin3}) gives
the $SDiff(T^2)$ 3-algebra (\ref{47}).
\begin{eqnarray}
\bullet \ \  &&{[}T_{i_{1}}, T_{i_{2}}, T_{i_{3}}, T_{i_{4}}{]}  \notag \\
&=&2\left( \cos \left( \pi \alpha \left( {i}_{1}\wedge i_{2}+{i}_{1}\wedge
i_{3}+{i}_{1}\wedge i_{4}+i_{2}\wedge i_{3}+i_{2}\wedge i_{4}+i_{3}\wedge {i}%
_{4}\right) \right) \right.   \notag \\
&&-\cos \left( \pi \alpha \left( {i}_{1}\wedge i_{2}+{i}_{1}\wedge i_{3}+{i}%
_{1}\wedge i_{4}+i_{2}\wedge i_{3}+i_{2}\wedge i_{4}-i_{3}\wedge {i}%
_{4}\right) \right)   \notag \\
&&-\cos \left( \pi \alpha \left( {i}_{1}\wedge i_{2}+{i}_{1}\wedge i_{3}+{i}%
_{1}\wedge i_{4}-i_{2}\wedge i_{3}+i_{2}\wedge i_{4}+i_{3}\wedge {i}%
_{4}\right) \right)   \notag \\
&&+\cos \left( \pi \alpha \left( {i}_{1}\wedge i_{2}+{i}_{1}\wedge i_{3}+{i}%
_{1}\wedge i_{4}-i_{2}\wedge i_{3}-i_{2}\wedge i_{4}+i_{3}\wedge {i}%
_{4}\right) \right)   \notag \\
&&+\cos \left( \pi \alpha \left( {i}_{1}\wedge i_{2}+{i}_{1}\wedge i_{3}+{i}%
_{1}\wedge i_{4}+i_{2}\wedge i_{3}-i_{2}\wedge i_{4}-i_{3}\wedge {i}%
_{4}\right) \right)   \notag \\
&&-\cos \left( \pi \alpha \left( {i}_{1}\wedge i_{2}+{i}_{1}\wedge i_{3}+{i}%
_{1}\wedge i_{4}-i_{2}\wedge i_{3}-i_{2}\wedge i_{4}-i_{3}\wedge {i}%
_{4}\right) \right)   \notag \\
&&-\cos \left( \pi \alpha \left( -{i}_{1}\wedge i_{2}+{i}_{1}\wedge i_{3}+{i}%
_{1}\wedge i_{4}+i_{2}\wedge i_{3}+i_{2}\wedge i_{4}+i_{3}\wedge {i}%
_{4}\right) \right)   \notag \\
&&+\cos \left( \pi \alpha \left( -{i}_{1}\wedge i_{2}+{i}_{1}\wedge i_{3}+{i}%
_{1}\wedge i_{4}+i_{2}\wedge i_{3}+i_{2}\wedge i_{4}-i_{3}\wedge {i}%
_{4}\right) \right)   \notag \\
&&+\cos \left( \pi \alpha \left( -{i}_{1}\wedge i_{2}-{i}_{1}\wedge i_{3}+{i}%
_{1}\wedge i_{4}+i_{2}\wedge i_{3}+i_{2}\wedge i_{4}+i_{3}\wedge {i}%
_{4}\right) \right)   \notag \\
&&-\cos \left( \pi \alpha \left( -{i}_{1}\wedge i_{2}-{i}_{1}\wedge i_{3}-{i}%
_{1}\wedge i_{4}+i_{2}\wedge i_{3}+i_{2}\wedge i_{4}+i_{3}\wedge {i}%
_{4}\right) \right)   \notag \\
&&-\cos \left( \pi \alpha \left( -{i}_{1}\wedge i_{2}+{i}_{1}\wedge i_{3}-{i}%
_{1}\wedge i_{4}+i_{2}\wedge i_{3}+i_{2}\wedge i_{4}-i_{3}\wedge {i}%
_{4}\right) \right)   \notag \\
&&\left. +\cos \left( \pi \alpha \left( -{i}_{1}\wedge i_{2}-{i}_{1}\wedge
i_{3}-{i}_{1}\wedge i_{4}+i_{2}\wedge i_{3}+i_{2}\wedge i_{4}-i_{3}\wedge {i}%
_{4}\right) \right) \right) T_{i_{1}+i_{2}+i_{3}+i_{4}}. \label{sine4}
\end{eqnarray}
For the cosine 4-algebra (\ref{sine4}) with the arbitrary value $\alpha\in R$,
it is a sh-4-Lie algebra.
Taking the rescaled generators
$\bar{T}_{i_{j}}=\sqrt[3]{\frac{-\mathbf{i}}{2\pi \alpha }}T_{i_{j}}$, $j=1,2,3,4$,
in (\ref{sine4}),
then when $\alpha\rightarrow 0$,
(\ref{sine4}) becomes the null 4-algebra
\begin{eqnarray}
[\bar{T}_{i_1}, \bar{T}_{i_2}, \bar{T}_{i_3}, \bar{T}_{i_4}]=0.
\end{eqnarray}

$\bullet $ \ \ When $n=5$, the corresponding 5-algebra is given by (\ref%
{sin5alg}). Taking $\alpha =\frac{1}{3}$ in (\ref{sin5alg}), it immediately
gives a sh-5-Lie algebra.
Let us take  rescaled generators 
$\bar{T}_{i_{j}}=\sqrt[4]{\frac{-\mathbf{i}}{2\pi \alpha }}T_{i_{j}}$, $j=1,2,\cdots ,5$,
in (\ref{sin5alg}), then we have  the null 5-algebra in the limit $\alpha\rightarrow 0$,
\begin{eqnarray}
[\bar{T}_{i_1}, \bar{T}_{i_2}, \bar{T}_{i_3}, \bar{T}_{i_4}, \bar{T}_{i_5}]=0.
\end{eqnarray}

Not as the case of the sine 3-algebra (\ref{Psin3}), we note that when $\alpha\rightarrow 0$,
the cosine 4-algebra (\ref{sine4}) and 5-algebra (\ref{sin5alg})
with respect to the  rescaled generators become the null 4 and 5-algebras, respectively.
For the (co)sine $n$-algebra (\ref{sinnalg}) with $n\ge 4$,
taking the  rescaled generators 
$\bar{T}_{i_{j}}=\sqrt[n-1]{\frac{-\mathbf{i}}{2\pi \alpha }}T_{i_{j}}$, $j=1,2,\cdots ,n$,
it is easy to verify that in the limit $\alpha\rightarrow 0$,
(\ref{sinnalg}) gives the null $n$-algebra (\ref{nnalg}).

\section{Concluding Remarks}

We have investigated the $q$-deformation of the infinite-dimensional $n$-algebras.
The V-W algebra is the centerless Virasoro algebra. Its $q$-deformation has
been well investigated in the literature. One has already known that in the
usual way, the V-W $n$-algebra is null. In this paper, we firstly investigated the
$q$-deformation of the null V-W $n$-algebra and constructed the nontrivial
$q$-deformed V-W $n$-algebra. We found that it satisfies the
sh-Jacobi's identity, but the FI fails. Thus this $q$-deformed
V-W $n$-algebra is indeed a sh-$n$-Lie algebra.
Furthermore in terms of the pseud-differential operators on the quantum plane,
we constructed the $q$-deformed $SDiff(T^2)$ $n$-algebra and proved that
the sh-Jacobi's identity holds for even $n$.
We also presented the (co)sine $n$-algebra which is the sh-$n$-Lie algebra
for the case of even $n$.
The interesting cases are for $n=3$ and $5$. We found that there exists a
sine 3-algebra  which is indeed a Fillipov 3-algebra.
When $n=5$, we derived a cosine 5-algebras which is a sh-5-Lie algebras.
An interesting open question is whether there exists the special values
such that the (co)sine $n$-algebra with general odd $n$ is the Fillipov $n$-algebra
or sh-$n$-Lie algebra.

It is worthwhile to mention that we introduce the appropriate star product
into the $q$-$n$-bracket. This star product plays an important role in
deriving the desired $q$-deformed V-W $n$-algebra.
For the case of the $q$-deformed $SDiff(T^2)$ $n$-algebra, we tried to
introduce the appropriate star product such that the sh-Jacobi's identity holds for odd $n$.
Unfortunately, we did not succeed in finding any one.
Whether there exists such kind of the star product still deserves further study.

Our investigation revealed a deep connection between the $q$-deformed
infinite-dimensional $n$-algebra and the sh-$n$-Lie algebra. One may
construct the sh-$n$-Lie algebra from the $q$-deformation point of view. It
sheds new light on the sh-$n$-Lie algebra. It would be interesting to study
further and see whether there exist the central extension terms for the
sh-$n$-Lie algebra derived in this paper.
Finally, it is worth to emphasize that we give
an explicit physical realization of the (co)sine $n$-algebra
in terms of the so-called magnetic translation operators. We believe that
the application of the $q$-deformed V-W and $SDiff(T^2)$  $n$-algebras
in physics should also be of interest.

\section*{Acknowledgements}

The authors are grateful to Morningside Center of Chinese Academy of
Sciences for providing excellent research environment and financial support
to our seminar in mathematical physics. This work is partially supported by
NSF projects (11375119 and 11475116), KZ201210028032.


\end{document}